\documentclass[preprint]{elsarticle}
\usepackage{algorithm, algorithmic, amsfonts,amsmath,amsthm}
\usepackage{comment}
\usepackage{color}
\usepackage[all]{xy}
\usepackage{subfigure}

\newtheorem{thm}{Theorem}
\newtheorem{cor}[thm]{Corollary}
\newtheorem{prop}[thm]{Proposition}
\newtheorem{lem}[thm]{Lemma}

\theoremstyle{definition}

\newtheorem{ex}[thm]{Example}

\newcommand{\rs}{\mathrm{rs}}
\newcommand{\F}{\mathbb{F}}
\newcommand{\G}{\mathcal{G}_q(k,n)}
\newcommand{\U}{\mathcal{U}}
\newcommand{\V}{\mathcal{V}}
\newcommand{\R}{\mathcal{R}}
\newcommand{\E}{\mathcal{E}}
\newcommand{\C}{\mathcal{C}}
\newcommand{\PG}{PG(n, q)}
\newcommand{\GL}{\mathrm{GL}}

\begin{document}

\title{Spread Decoding in Extension Fields}
\author{Felice Manganiello\fnref{ref1}}
\address{Department of Electrical and Computer Engineering, University of Toronto}
\author{Anna-Lena Trautmann\fnref{ref2}}
\address{Institute of Mathematics, University of Zurich}
\fntext[ref1]{The author is partially supported by Swiss National Science Foundation Grant no. 126948 and 135934.}
\fntext[ref2]{The author is partially supported by Swiss National Science Foundation Grant no. 126948 and 138080.}
\ead[url]{www.math.uzh.ch/aa}

\begin{abstract}
  A spread code is a set of vector spaces of a fixed dimension over a
  finite field $\F_q$ with certain properties used for random network
  coding. It can be constructed in different ways which lead to
  different decoding algorithms. In this work we consider one such 
  representation of spread codes and present a minimum distance
  decoding algorithm which is efficient when the codewords,
  the received space and the error space have small dimension.
\end{abstract}

\maketitle

\section{Introduction}

In network coding, one is interested in efficient communication
between different sources and receivers in a network which is
representable through a directed acyclic graph. In multicast, one is
looking at the communication between a sender to several receivers,
where each receiver should receive the message sent by the
sender. In \cite{ko03a,li03} it is proven that one achieves the
communication rate simply by allowing nodes of the network to forward
random linear combinations of its information vectors. If the
underlying topology of the network is unknown we speak about
\textit{random linear network coding}. Since linear spaces are
invariant under linear combinations, they are what is needed as
codewords \cite{ko08}. It is helpful for decoding to constrain
oneself to subspaces of a fixed dimension, in which case we talk about
\emph{constant dimension codes}.

One class of constant dimension codes is the one of \emph{spread
  codes}. These codes have maximal minimum distance and are optimal in
the sense that they achieve the Singleton-like bound on the
cardinality of network codes. They can be constructed with the help of
companion matrices of irreducible polynomials, as explained in
\cite{ma08p}.

In this work we translate the construction of \cite{ma08p} to an extension field
setting and evolve a minimum distance decoding algorithm for spread
codes in this setting. The complexity of this new algorithm depends on
different parameters than the algorithms of \cite{go11,ko08,si08a},
which are also applicable to spread codes. Therefore, depending on the
network setting applied on, the new algorithm has an improved
performance.

The paper is structured as follows: In Section \ref{2} we give some
preliminaries on random network coding and constant dimension
codes. The main results of this work are found in Section \ref{3},
where we first show how to translate the spread code construction of \cite{ma08p} into
a different setting and then explain how decoding can be done in this
setting. We study the complexity of the decoding algorithm and give
comparison to other known decoding algorithms in Section
\ref{4}. Moreover, we study the probability that the algorithm
terminates after fewer steps than the worst case scenario. We conclude
this work in Section \ref{5}.


\section{Preliminaries}\label{2}

Let $\mathbb{F}_q$ be the finite field with $q$ elements, where
$q$ is a prime power. We  denote the set of all subspaces of
$\F_{q}^{n}$ by $\PG$ and the set of all $k$-dimensional subspaces of
$\F_q^n$, called the Grassmannian, by $\G$. The general linear group
$GL_n$ is the set of all invertible $n\times
n$-matrices with entries in $\F_{q}$. Moreover, the set of all
$k\times n$-matrices over $\F_q$ is denoted by $Mat_{k\times n}$.

Let $U\in Mat_{k\times n}$ be a matrix of rank $k$ and
\[
\mathcal{U}=\rs (U):= \text{row space}(U)\in \G.
\]
One can notice that the row space is invariant under $GL_k$-multiplication from
the left, i.e. for any $T\in GL_k$
\[
\mathcal{U}=\rs(U)= \rs(T U).
\]
Thus, there are several matrices that represent a given subspace. A
unique representative of these matrices is the one in reduced row
echelon form.
Any $k\times n$-matrix can be transformed into reduced row echelon
form by a $T\in GL_k$.

A \textit{subspace code} is simply a subset of $\PG$ and a \textit{constant dimension code} is a subset of
the Grassmannian $\G$. 

The \emph{subspace distance}, given by
\begin{align*}
  d_S(\mathcal{U},\mathcal{V}) =& \dim(\U +\V) - \dim(\mathcal{U}\cap
  \mathcal{V})
\end{align*}  
for $\U,\V$ two subspaces of $\F_q^n$, is a metric function on $\PG$. It
induces a metric on $\G$ by
\begin{align*}
  d_S(\mathcal{U},\mathcal{V}) =& 2k - 2\dim(\mathcal{U}\cap
  \mathcal{V})
  \end{align*}  
for any $\mathcal{U},\mathcal{V} \in \G$. The minimum distance of a subspace code $\C \subseteq \PG$ is defined as
\[d(\C):=\min\{d_{S}(\U,\V) \mid \U,\V \in \C, \ \U\neq\V\}.\] 
The subspace distance is a suitable distance for coding over
the operator channel \cite{ko08}, where errors and erasures can be corrected. An error corresponds to an inserted erroneous vector, i.e. an increase in dimension, whereas an erasure is a decrease in dimension of the code word.  
The error-and-erasure correction capability of a code $\C \subseteq \PG$ with minimum distance $d(\C)$ is
\[t:=\left\lfloor\frac{d(\C)-1}{2}\right\rfloor.\] 
Different
constructions for constant dimension codes can be found e.g. in
\cite{et08u, ko08p, ko08, si08a, tr10p}.

In the case that $k$ divides $n$ one can construct codes with minimum
distance $2k$ and cardinality $\frac{q^{n}-1}{q^{k}-1}$, called
\textit{spread codes}. One construction for spread codes
is the following \cite{ma08p}. Let $p(x)=\sum_{i=0}^{k} p_{i}x^{i} \in \F_q[x]$ be
an irreducible monic polynomial of degree $k$. Its companion matrix
$P\in Mat_{k\times k}$ is
  \[P=\left(\begin{array}{cccccc}
  0 & 1 &0 & \dots & 0\\
  0 & 0 &1 &          & 0\\
  \vdots &&&\ddots\\
  0 & 0 &0 &          & 1\\
  -p_{0} & -p_{1}& -p_{2} &\dots & -p_{k-1}
  \end{array}\right) .\] 
Then the $\F_{q}$-algebra of such a companion matrix $\F_{q}[P]$ is a finite field and  
the set
\[
\left\{\rs \left[\begin {array}{cccc}
        B_0 & B_1 & \cdots & B_{\frac{n}{k}-1}
       \end {array}\right] \in \G \mid B_i\in \F_q[P]\right\}
\]
is a spread code. Since we want to work with the reduced row echelon
form as a unique matrix representation of the vector spaces we assume
that the first non-zero block from the left is the identity matrix.
These codes are optimal because they achieve the Singleton-like bound \cite{ko08}
on the cardinality of constant dimension codes for a given minimum distance.


\section{Spread Codes in Extension Field Representation}\label{3}

\subsection{Translation of the Construction}

We now translate the construction of spread codes of \cite{ma08p} from
the companion matrix to an extension field setting. Spreads of this
type are also known as \emph{$\F_{q}$-linear representations of
  $PG(l,{q^{k}})$} \cite{ba11a} or \emph{Desarguesian $(k-1)$-spreads}
\cite[p. 12]{la01t}. Since they exist for any degree over $\F_{q}$, we
choose primitive polynomials and their companion matrices for the
spread code constructions.  

For the remain of this paper assume that $k|n$ and let
$l=n/k$. Moreover, let $\alpha\in \F_{q^k}$ be a primitive element of $\F_{q^k}$
and $\beta\in \F_{q^n}$ a primitive element of $\F_{q^n}$ as an extension
field of $\F_{q^k}$. 
The polynomial $p(x)\in\F_q [x]$ denotes the minimal polynomial of
$\alpha$ and $P \in GL_{k}$ denotes its companion
matrix. It holds that $\mathrm{ord}(P)=q^{k}-1$ \cite[Lemma 6.26]{li94}.

Denote by $\phi^{(k)}: \F_q^k\rightarrow
\F_{q^k}$ and $ \phi^{(l)}:\F_{q^k}^l \rightarrow \F_{q^n}$ the standard
vector space isomorphisms:
\begin{align*}
 \phi^{(k)}(u_1,\dots,u_{k}) &= \sum_{i=0}^{k-1} u_{i+1} \alpha^{i} ,\\
 \phi^{(l)}(v_1,\dots,v_{l}) &= \sum_{i=0}^{l-1} v_{i+1} \beta^{i} ,
\end{align*}
for $(u_1,\dots,u_{k}) \in \F_{q}^{k}$ and $(v_1,\dots,v_{l}) \in F_{q^{k}}^{l}$.

\begin{prop} 
  Let $e_1,\dots, e_n$ be the standard basis of $\F_q^n$. Then
\[\bigcup_{i=0}^{l-1} \{\beta^i,\alpha \beta^i ,\dots, \alpha^{k-1} \beta^i\}\]
is a basis of $\mathbb{F}_{q^n}$ over $\F_q$ and
\begin{align*}
\phi: \mathbb{F}_q^n &\longrightarrow \mathbb{F}_{q^n}\\ 
  e_i &\longmapsto \alpha^{(i-1 \mod k)} \beta^{\lfloor \frac{i-1}{k}\rfloor}
\end{align*}
is a vector space isomorphism. 
\end{prop}
\begin{proof}
Define $\tilde{\phi}^{(k)} : \F_{q}^n \rightarrow \F_{q^{k}}^l$,
\[\tilde{\phi}^{(k)}(v_1,\dots,v_n):=\left(\phi^{(k)}(v_1,\dots,v_k),\dots,
\phi^{(k)}(v_{n-k+1},\dots,v_n)\right).\]
Then $\phi^{(l)}, \tilde{\phi}^{(k)}$ are vector space isomorphisms and $\phi = \phi^{(l)}\circ \tilde{\phi}^{(k)}$ satisfies the following diagram
\[
\xymatrix{
\F_{q}^n \ar[dr]_{\tilde{\phi}^{(k)}} \ar[rr]^\phi & & \F_{q^n}\\  
&\F_{q^k}^l\ar[ur]_{\phi^{(l)}}&
}
.\]

\end{proof}

Note, that $\phi$ can be applied on sets of vectors (e.g. vector spaces) element-wise.

\begin{lem}\label{lem2}
Denote by $P[i]$ the $i$-th row vector of $P$. Then
\[\phi^{(k)} (P^h[i]) = \alpha^{h+i-1} \]
for $i=1,\dots, k$ and $h=1,\dots,q^{k}-1$.
\end{lem}

\begin{proof}
It is easy to see that $\phi^{(k)}(P[i])=\alpha^{i}$ for $i\in \{1,\dots,k\}$.  
Moreover, $\phi^{(k)}$ is commutative with the multiplication with $P$ and $\alpha$:
\begin{align*}
\phi^{(k)}(uP)=& \sum_{i,j=1}^ku_{i }P_{ij}\alpha^{j-1}\\
=& \sum_{i=1}^ku_{i}\alpha^{i} \left(\sum_{j=1}^k P_{ij} \alpha^{j-i-1}\right)\\
=& \sum_{i=1 }^k u_{i} \alpha^{i} = \phi^{(k)}(u)\alpha
\end{align*}
for all $u=(u_1,\dots,u_{k}) \in \F_{q}^{k}$. 
\[\implies \phi^{(k)} (P^h[i]) = \phi^{(k)} (P[i]P^{h-1})= \phi^{(k)}
(P[i])\alpha^{h-1}= \alpha^{h+i-1} \]
\end{proof}

Recall that the spread code elements are of the type
\[\rs \left[\begin{array}{cccc} B_0 & B_1 & \dots &
    B_{l-1} \end{array}\right] \quad \in \G\]
where each block $B_i$ is an element of $\F_q[P]$ and the first non-zero block from the left is the identity.

\begin{thm}\label{t:unique_ident}
Define
\[\gamma_{j}:=\left\{\begin{array}{ll} 0 & \textnormal{ if }
    B_j=0 \\ \alpha^h & \textnormal{ if }
    B_j=P^h \end{array}\right.  \in \F_{q^{k}} .\]
Then
\[\phi(\rs \left[\begin{array}{cccc} B_0 & B_1 & \dots & B_{l-1} \end{array}\right]) = \F_{q^k} \cdot \sum_{j=0}^{l-1} \gamma_j \beta^j .\]
Hence, we can uniquely identify each spread code element by the
respective $\gamma=(\gamma_0,\dots,\gamma_{l-1})\in \F_{q^k}^l$.
\end{thm}

\begin{proof}
Denote by $B_{j}[i]$ the $i$-th row vector of the block $B_j$. From Lemma \ref{lem2} we know that $\phi^{(k)}(B_{j}[i])=\alpha^{i-1} \gamma_j$.
The power of $\beta$ corresponds to the position of the block $B_j$, thus in general $\phi$ maps the $i$-th row of the whole matrix to 
\[\sum_{j=0}^{l-1} \alpha^{i-1}\gamma_{j} \beta^j =  \alpha^{i-1}\sum_{j=0}^{l-1}\gamma_{j} \beta^j \quad \forall i=1,\dots,k .\] 
As $\alpha$ is a primitive element of $\F_{q^k}$, the elements of the vector space are exactly mapped to $\F_{q^k} \cdot \sum_{j=0}^{l-1}\gamma_{j} \beta^j$.
\end{proof}

\begin{cor}
A spread code $\C \subseteq \G$ constructed as before is isomorphic to $\mathcal{G}_{q^{k}}(1,l)$.
\end{cor}
\begin{proof}
Since a spread code covers the whole space and $\phi$ maps a code word to an $\F_{q^{k}}$-linear subspace with basis vector $\gamma$, the statement holds.
\end{proof}


\subsection{Decoding}

One can now use this structure for the decoding procedure in this representation. 

First assume only erasures and no errors happened during transmission. Then any received vector space $\R$ with $\dim(\R)\geq 1$ can be decoded to its closest code word, since the number of errors and erasures is less than or equal to $k-1= \frac{2k-2}{2}= t$. For decoding choose an element of the received space $r \in \R$ and compute $\gamma=(\gamma_0,\dots,\gamma_{l-1}) \in \F_{q^{k}}^{l}$ such that
\[\phi(r)= \alpha^{i-1}\sum_{j=0}^{l-1}\gamma_{j} \beta^j\]
for some $i$. For this, divide $r$ into $l$ blocks of size $k$,
$r_{1},\dots,r_{l}$, and find the first non-zero block, denoted by
$r_{s}$. It holds that $r_{s}={\phi^{(k)}}^{-1}(\alpha^{i-1})$, since
the first non-zero block is the identity matrix in the
construction. Then $\gamma$ can be computed by at most one inversion
and $l$ multiplications in $\F_{q^{k}}$ since
\[\gamma= (\phi^{(k)} (r_{1}) \phi^{(k)} (r_{s})^{-1} ,  \dots  , \phi^{(k)} (r_{l}) \phi^{(k)}(r_{s})^{-1} ).\]
Note, that one does not need to compute the discrete logarithm to find the $i$. It is enough to compute $\gamma$ to identify the code word.

\begin{ex}
Let $p(x)=x^3+x+1 \in \F_2[x]$ and consider   
  the spread code of constant dimension $3$ over $\F_2^6$ generated by it. Let
  $r=(110 | 101)$ be a received vector. It holds that
  $\phi(r)=1+\alpha + \beta + \alpha^2 \beta$.  The first three vector
  entries tell you that you have to divide by $1+\alpha$ to compute
  $\gamma$:
\[(1+\alpha + \beta + \alpha^2 \beta) {(1+\alpha)}^{-1} = 1+(1+\alpha) \beta\]
Hence, $\gamma= (1, 1+\alpha) $, which identifies the codeword 
\[ \rs\left[\begin{array}{c|c}
       I & {I+ P}
      \end{array}\right] =
\rs\left[\begin{array}{ccc|ccc}
       1&0&0&1&1&0\\
       0&1&0&0&1&1\\
       0&0&1&1&1&1   
      \end{array}\right]
.\]
\end{ex}

But what if errors were inserted? 
Let $\U \in \C$ be the sent code word and denote by $k'$ the dimension of the received vector space $\R \in \PG$. 
For correct decoding it has to hold that 

\begin{eqnarray*}
&d_{S} (\U, \R) \leq \lfloor\frac{d(\C)-1}{2}\rfloor \\
\iff &k+k' - 2\dim(\U \cap \R) \leq k-1\\
\iff &   \dim(\U \cap \R) \geq \frac{k'+1}{2} ,
\end{eqnarray*}

therefore one needs to find $\lceil\frac{k'+1}{2}\rceil$ linearly independent elements of $\R$ with the same respective $\gamma$, called $\gamma_{\max}$. Then we decode to the codeword
\[\phi^{-1}(\F_{q^k}\cdot \sum_{j=0}^{l-1} {\gamma_{\max}}_j \beta^j ).\]
Since we do not know if any or which of the elements of $\R$ are erroneous, one needs to examine not only a basis but all elements of $\R$.

A first basic decoding algorithm in this extension field representation is given in Algorithm \ref{alg1}.
All field operations are done over $\mathbb{F}_{q^k}$.

\begin{algorithm}
\caption{Basic decoding algorithm.}
\label{alg1}                       
\begin{algorithmic}             
\REQUIRE the received vector space $\R \in \PG$, $k'=\dim(\R)$

\FOR{each $v \in \R$}
\STATE divide $v$ into blocks $v_0,\dots, v_{l-1}$ of length $k$
\STATE $v_s:=$ the first block from the left with non-zero entries
\STATE $a:=(\phi^{(k)}(v_s))^{-1}$ 
\STATE store $\gamma_{v}:=(\phi^{(k)}(v_0)\cdot a,\dots, \phi^{(k)} (v_{l-1})\cdot a)$
\ENDFOR
\STATE $\gamma_{\max} :=$ the element of highest multiplicity in $\{\gamma_{v} | v \in \R\}$
\IF{there are $ \geq\lceil \frac{k'+1}{2} \rceil$ linearly independent $v\in \R$ such that $\gamma_{v}=\gamma_{\max}$} 
\RETURN $\phi^{-1}(\F_{q^k}\cdot \sum_{j=0}^{l-1} {\gamma_{\max}}_j \beta^j)$
\ELSE 
\RETURN ``not decodable''
\ENDIF
\end{algorithmic}
\end{algorithm}


\subsection{Improvements on the Algorithm}

We improve the algorithm by systematically choosing the linear combinations of the basis vectors of the received space to work with. 
For it, note that errors are canceled out in some linear combinations of elements, as illustrated in the following example.

\begin{ex}
Assume $\U \in \C$ was sent and consider two elements of the received space $r_{1}, r_2 \in\R$ containing the same error $e \in \F_{q}^{n}$, i.e.
\[r_{1}=\sum_{u\in \U} \lambda_{u} u +e  \quad, \quad r_{2}=\sum_{u\in\U} \mu_{u} u + e \]
for some $\lambda_u, \mu_u \in \F_q$. Then 
\[r_{1}+(q-1)r_{2} = \sum_{u\in \U} \lambda_{u} u +e -\sum_{u\in \U} \lambda_{u} u -e =\sum_{u\in \U} (\lambda_{u}-\mu_{u}) u  \quad \in \U .\]
\end{ex}

Let us generalize this idea to arbitrary numbers of errors.

\begin{prop}\label{p5}
  Let $u_{1},\dots,u_{k} \in \F_{q}^{n}$ be a basis of the sent code
  word $\U \in \G$ and $r_{1},\dots, r_{k'} \in \F_{q}^{n}$ a basis of
  the received space $\R$. Assume $f<k'$ linearly independent error
  vectors were inserted during transmission, i.e. $\R= \U' \oplus
  \mathcal{E}$, where $\U'$ is a subspace of $\U$ and $\mathcal{E}$ is
  the vector space of dimension $f$ spanned by the error vectors. Then
  the set
\[\left\{\sum_{i\in I} \lambda_{i} r_{i} \mid  \lambda_{i} \in \F_{q},
  I\subset \{1,\dots,k'\}, |I|=f+1 \right\} \]
contains $k'-f$ linearly independent elements of $\U$.
\end{prop}
\begin{proof}
 Inductively on $f$: 
\begin{enumerate}
 \item 
If $f=0$, then $r_{1},\dots,r_{k'} \in\U$.
 \item
If $f=1$, assume $r_1,\dots,r_l  \not\in \mathcal{U}'$ and $r_{l+1},\dots,r_{k'}\in \U'$. Then there exist $\lambda_{ij}\in\F_q, \mu_i \in \F_q\setminus \{0\}$ such that
\[r_i = \sum_{j=1}^k \lambda_{ij} u_j + \mu_i e \quad \forall \; i=1,\dots,l\]
where $e \in \mathcal{E}$ denotes the error vector. Hence $\forall \; i,h=1,\dots,l$
\[r_i + r_h = \sum_{j=1}^{k} (\lambda_{ij}+\lambda_{hj}) u_j + (\mu_i
+\mu_h)e \]
\[\implies r_i + (-\mu_i \mu_h^{-1}) r_h = \sum_{j=1}^k (\lambda_{ij}-\mu_i\mu_h^{-1}\lambda_{hj}) u_j \quad \in \U'.\]
Then the elements $r_{l+1},\dots,r_{k'}, r_1 + (-\mu_1 \mu_2^{-1})
r_2,\dots,r_1 + (-\mu_1 \mu_l^{-1}) r_l$ are $k'-1$ linearly
independent elements without errors.
 \item If more errors, say $e_{1},\dots,e_{f}$, were inserted, then one can inductively ``erase'' $f-1$ errors in the linear combinations of at most $f$ elements. Write the received elements as
 \[r_i = \sum_{j=1}^k \lambda_{ij} u_j + \sum_{j=1}^f \mu_{ij} e_{j} \quad \forall \; i=1,\dots,k' \]
 with $\lambda_{ij}, \mu_{ij} \in \F_q$.
 Assume $\mu_{1f},\dots,\mu_{lf} \neq 0$ and $\mu_{(l+1)f},\dots,\mu_{k'f}=0$, i.e. the first $l$ elements involve $e_{f}$ and the others do not. 

 From above we know that the linear combinations of any two elements of $r_{1},\dots,r_{l}$  include $l-1$ linearly independent elements without $e_{f}$. Denote them by $m_{1},\dots,m_{l-1}$. Naturally these elements are also linearly independent from $r_{l+1},\dots,r_{k'}$. Use the induction step on $m_{1},\dots,m_{l-1},r_{l+1},\dots,r_{k'}$ to get $k'-1-(f-1)=k'-f$ linearly independent elements without errors.
\end{enumerate}
\end{proof}

\begin{cor}
In the setting of Proposition \ref{p5} assume $d_{S}(\R,\U)\leq t$, i.e. $\R$ is decodable. Then there are at least $\lceil\frac{k'+1}{2}\rceil$ linearly independent elements of $\U$ in the set
\[\mathcal{L} := \left\{\sum_{i\in I} \lambda_{i} r_{i} \mid  \lambda_{i} \in \F_{q},
  I\subset \{1,\dots,k'\}, |I|=f+1 \right\} .\]
\end{cor}
\begin{proof}
Let $\bar{f}$ denote the number of erasures. Then $\bar{f}=f+k-k'$ and thus
\[f+\bar{f} \leq t \iff 2f+k-k'\leq k-1 \iff f\leq \frac{k'-1}{2} .\]
With Proposition \ref{p5} it follows that $\mathcal{L}$ contains $k'-f\geq \frac{k'+1}{2}$ linearly independent vectors of the sent vector space.
\end{proof}

We use this fact to modify Algorithm \ref{alg1} as follows: We choose
a basis $r_{1},\dots,r_{k'}$ of the received space $\R\in \PG$ and
compute $\gamma_{r_{i}}$ for $i=1,\dots,k'$. Then we compute the
respective $\gamma$ of all linear combinations of two basis elements,
then of three elements etc. As before we can stop the process and
decode to a code word as soon as we have more than or equal to $\lceil \frac{k'+1}{2} \rceil$ linearly independent elements with the same $\gamma$.
This way, if $f$ errors occurred, we do not have to consider all elements of $\R$ but only the linear combinations of at most $f+1$ of the basis vectors.

Moreover note, that
a linear combination of elements with the same $\gamma$ is always
another element with $\gamma$. Since we need to find linearly
independent elements, it is therefore enough to check only
combinations of elements with different respective $\gamma$'s.

 It is possible to further improve the algorithm by restricting the
 elements of the basis which are used in Proposition \ref{p5}. 

\begin{lem}
  Let $\U\in \C$ and $\R\in PG(n,q)$ of dimension $k'$ such that
  $d_S(\U,\R)\leq k-1$. Let $B=\{r_1,\dots,r_{k'}\}$ be a
  basis of the received space $\R$ such that the matrix obtained by
  stacking the basis elements is in reduced row echelon form. Let
  $r_i=(r_{i,0},\dots,r_{i,l-1})\in (\F_q^k)^l$, where $kl=n$. Consider the
  partition of $B$ into the subsets
  \[B_j=\{r_i\mid r_{i,t}=0 \ \forall t<j \mbox{ and } r_{i,j}\neq
  0\}\]
  for $j=0,\dots,l-1$.
  Then, \[\U\cap\R\subseteq \langle \bigcup_{j=j'}^{l-1}B_j\rangle\] where $j'$ is such that
  $|B_{j'}|\geq\lceil \frac{k'-1}2\rceil$.
\end{lem}

\begin{proof}
  From Theorem \ref{t:unique_ident}, let $\gamma=(\gamma_0,\dots,
  \gamma_{l-1})$ be the unique identifier of the sent space $\U$ and
  $j':=\min\{j\in 0,\dots,l-1 \mid \gamma_j\neq 0\}$. It follows that
  the first $j'-1$ $k$-tuples of coordinates of each element of
  $\U\cap \R$ are zeros. Combining this with the properties of a reduced row echelon form, it follows that $\U\cap \R$ intersects
  trivially with $\langle\bigcup_{j=0}^{j'-1}B_j\rangle$. Then
  $\U\cap \R \subseteq\langle \bigcup_{j=j'}^{l-1}B_j\rangle$. 
Since $\dim(\U\cap \R)\geq \lceil \frac{k'-1}2\rceil$ it follows that 
$|B_{j'}|\geq\lceil \frac{k'-1}2\rceil$. 
\end{proof}

A consequence of the previous lemma is that the elements of the basis
contained in $\langle\bigcup_{j=0}^{j'-1}B_j\rangle$ are actually
erroneous. Thus, the algorithm can be altered such that it only works with the basis vectors inside $\bigcup_{j=j'}^{l-1}B_j$, which again improves the performance of the algorithm.



\section{Performance of the Algorithm}\label{4}
\subsection{Complexity}\label{4.1}

For a better understanding of the complexity of the algorithm we first
consider binary spread codes and then generalize it. Note, that the algorithm works for received spaces of arbitrary dimension.

If $k'$ is the dimension of the received space and $f$ is the dimension of the error space $\mathcal{E}$, the algorithm computes the
sums of at most $f+1$ basis vectors, which are
$\binom{k'}{f+1}$ many. For each sum it proceeds with an inversion
and at most $\frac{n}{k}-1$ multiplications over $\F_{2^{k}}$. The
complexity of inverting is upper-bounded by $\mathcal{O} (k^{2})$ over
$\F_2$ and the one of multiplying by $\mathcal{O} (k \log k)$ over
$\F_2$ using the FFT \cite[Chapter 8.2]{ga03}. Using the
approximation $\binom{k'}{f+1}\approx \frac{k'^{f+1}}{(f+1)!}$, the overall 
complexity is upper-bounded by $\mathcal{O} (nkk'^{f+1})$ over $\F_{2}$.

Over $\F_q$ one needs to consider not only sums but $\F_q$-linear
combinations. Thus we get an upper bound of $\binom{qk'}{f+1}$
combinations to check.  Hence, the  overall complexity is
upper-bounded by $\mathcal{O} (nk(qk')^{f+1})$ over $\F_{q}$.

The complexity reduces when some of the generators of the sent
codeword are not influenced by the errors since in this case the algorithm has to check only linear combinations of a smaller amount of basis vectors of the
received space.

In the following we compare this complexity with the one of the spread decoding algorithm shown in
\cite{go11} and the decoding algorithms for Reed-Solomon like codes
contained in \cite{ko08} and \cite{si08a} in the case of $q=2$ and $k=k'$.  
In \cite{go11} the authors
present a minimum distance decoder for the their spread code
construction. The complexity of their algorithm is
$\mathcal{O}((n-k)k^3)$. If the dimension of the error space is
minimal the two algorithms perform the same. 
When applied to
spread codes the complexities of the algorithms presented in \cite{ko08} and
\cite{si08a} are $\mathcal{O}(n^2(n-k)^2)$ and $\mathcal{O}(k(n-k)^3)$, respectively. The algorithm proposed in this work performs better if the
dimension of the codewords, of the received space and of the error
space are small.


\subsection{Probability of Better Performance}\label{4.2}

The aforementioned complexity considers the worst case scenario of a decoding procedure. We will now investigate the expected amount of computations needed for the algorithm under the assumption that the channel transfer matrices are uniformly distributed. 
Usually, when we sent a codeword $\U$ of dimension $k$, received a codeword $\R$ of dimension $k'$ and $e$ many insertions were made during transmission, we model the transmission by
$$\R = \bar{\U} \oplus \E$$
where $\bar{\U}$ is a subspace of $\U$ and $\dim \E = e$. Since we use a matrix channel, the actual sent and received matrices are
$$R = A\left[ \begin{array}{c}\bar{U} \\ E\end{array} \right]$$
where $\rs (R) = \R, \rs (\bar{U})=\bar{\U}, \rs (E) =\E$ and $A\in \GL_{k'}$ is the channel transfer matrix representing the random linear combinations done throughout the whole network.


For simplicity  we assume that $\dim \E = 1$ and compute in the following the probabilities that the algorithm terminates after the first round, i.e. after only considering the received vectors and no linear combinations of them. 

\begin{lem}
The set of elements of $\GL_{k}$  whose last column has $z$ many zero entries has cardinality
\[(q-1)^{k-z} \binom{k}{z} \prod_{i=1}^{k-1}(q^{k}-q^{i}) .\]
\end{lem}
\begin{proof}
First we compute how many $v\in \F_{q}^{n}$ with exactly $z$ zeros exist. Fix the first $z$ positions to be zero, then there are $k-z$ positions to be filled with non-zero elements. Thus, there are $(q-1)^{k-z}$ possible vectors. Moreover, we have $\binom{k}{z} $ many possibilities to choose different zero positions, hence there are $(q-1)^{k-z} \binom{k}{z} $ many $v\in \F_{q}^{n}$ with exactly $z$ zeros. We fix one of these vectors as the last column of our $\GL_{k}$ element, then the next column can be chosen from $q^{k}-q$ elements etc. 
\end{proof}

\begin{thm}
The probability that $z$ many of the received basis vectors $r_{1},\dots, r_{k'}$ are not influenced by the error is 
\[\frac{(q-1)^{k'-z}}{q^{k'}-1}  \binom{k'}{z} .\]
\end{thm}
\begin{proof}
If one error and one erasure occurred during transmission we can model this as an error-free transmission of the matrix where the last row is the error vector and the other rows are the basis vectors of the code word without the erasure, and the channel action is represented by $\GL_{k'}$-multiplication on the left. Then a received vector is a linear combination of elements including the error if and only if the respective position in the last column of the $\GL_{k'}$-element is non-zero. Thus, we divide the number of elements with $z$ many zeros by $|\GL_{k'}|$.
\end{proof}

%

\begin{thm}\label{thm10}
The expected number of received basis vectors that are error-free is
\[ \frac{k'(q^{k'-1}-1)}{q^{k'}-1} \approx \frac{k'}{q}  .\]
\end{thm}
\begin{proof}
Since one error occurred, at least one of the received vectors has to be erroneous. Then the expected value is
 \begin{align*}
  & \sum_{z=0}^{k'-1}z \frac{(q-1)^{k'-z}}{q^{k'}-1}  \binom{k'}{z} \\
=& \frac{(q-1)^{k'}}{q^{k'}-1} \sum_{z=0}^{k'-1} \frac{z}{(q-1)^{z}}  \binom{k'}{z} \\
=& \frac{(q-1)^{k'}}{q^{k'}-1} \frac{k'(q^{k'-1}-1)}{(q-1)^{k'}}\\
=&\frac{k'(q^{k'-1}-1)}{q^{k'}-1} .
 \end{align*}

\end{proof}

\begin{figure}[h]
\begin{center}
$
\begin{array}{|c|c|c|c|c|c||c|}
\hline  q\backslash k & 2 & 3& 4 &5 &6 & \approx\\
\hline  2 & \frac{2}{3} & \frac{9}{7}& \frac{28}{15} & \frac{75}{31} & \frac{62}{21} &\frac{k}{2}\\
\hline  3 & \frac{1}{2} & \frac{12}{13}& \frac{13}{10} & \frac{200}{121} & \frac{363}{182} &\frac{k}{3}\\
\hline  4 & \frac{2}{5} & \frac{5}{7}& \frac{84}{85} &\frac{425}{341} &\frac{682}{455} &\frac{k}{4}\\
\hline  5 & \frac{1}{3} & \frac{18}{31}& \frac{31}{39} &\frac{780}{781} &\frac{781}{651} &\frac{k}{5}\\\hline
\end{array}
$
\caption{Values for Theorem \ref{thm10}.}
\end{center}
\end{figure}


Because one needs more than $k'/2$ error-free basis vectors to decode correctly it follows that:

\begin{cor}\label{cor12}
If $\dim(\E)=1$ and the received space $\R$ is decodable, then the probability that the decoding algorithm terminates after the first round (i.e. after $nkk'$ operations over $\F_{q}$) is
\[1 - \frac{(q-1)^{k'} (q^{\lceil \frac{k'+1}{2}\rceil}-1)}{(q-1)^{\lceil \frac{k'+1}{2}\rceil}( q^{k'}-1)} \approx 1-\left(\frac{q-1}{q}\right)^{\lfloor\frac{k'-1}{2}\rfloor}.\]
\end{cor}
\begin{proof}
Let $l:=\lceil \frac{k'+1}{2}\rceil$. Then the probability that at least $l$ many of $r_{1},\dots, r_{k'}$ are error-free is
\begin{align*}
& \sum_{z=l}^{k'-1} \frac{(q-1)^{k'-z}}{q^{k'}-1}  \binom{k'}{z} \\
=&\frac{(q-1)^{k'}}{q^{k'}-1} \sum_{z=l}^{k'-1} \frac{1}{(q-1)^{z}}  \binom{k'}{z} \\
=&\frac{(q-1)^{k'}}{q^{k'}-1} \left(\frac{q^{k'}-1}{(q-1)^{k'}} - \frac{q^{l}-1}{(q-1)^{l}} \right) \\
=&1 - \frac{(q-1)^{k'} (q^{l}-1)}{(q-1)^{l}( q^{k'}-1)} 
\end{align*}
\end{proof}

\begin{figure}[h]
\begin{center}
$
\begin{array}{|c|c|c|c|c|c|c|c|}
\hline  q\backslash k' & 2 &  3 &4& 5 & 6 & 7 \\
\hline  2 & 0 & \frac{4}{7} & \frac{8}{15}  & \frac{24}{31} & \frac{16}{21} & \frac{112}{127}\\
\hline  3 & 0 & \frac{5}{13} & \frac{7}{20} & \frac{69}{121} & \frac{51}{91} & \frac{773}{1093}\\
\hline  4 & 0 & \frac{2}{7} & \frac{22}{85}& \frac{152}{341} & \frac{40}{91} & \frac{3166}{5461}\\
\hline
\end{array}
$
\caption{Values for Corollary \ref{cor12}.}
\end{center}
\end{figure}


Thus, one can see that if $k'\geq 2q-1$, then with probability greater than $0.5$ the algorithm terminates after the first round.

In a similar manner one can compute the same probability under the assumption that $\dim \E \geq 2$. Moreover, one can determine the probability that the algorithm terminates after the second round, the third round etc. 

\section{Conclusions}\label{5}

In this work we consider a certain construction for spread codes, which are
codes with optimal cardinality and error correction capability. The
construction is based on the representation of code words via a
unique element $\gamma\in\F_{q^k}^{n/k}$. We present a minimum distance
decoding algorithm which works by finding this unique $\gamma$. The
performance of the algorithm is mainly based on two properties: the
operations are done over $\F_{q^k}$ instead of $\F_{q^n}$ and the
elements tested for finding the $\gamma$ are only the linear combinations of $f$ of the basis vectors of the received space instead of
all elements of the vector space. As a result we obtain a decoding algorithm
with a good performance when the dimension of the codewords, of the received
space and of the error space are small.


\section*{Acknowledgement}

The authors thank Wolfgang Willems for the useful discussion during his visit and Joachim Rosenthal for the comments on the final version of this work.


\bibliography{/home/b/trautman/Dropbox/my_bib/network_coding_stuff}

\providecommand{\bysame}{\leavevmode\hbox to3em{\hrulefill}\thinspace}
\providecommand{\MR}{\relax\ifhmode\unskip\space\fi MR }
\providecommand{\MRhref}[2]{%
  \href{http://www.ams.org/mathscinet-getitem?mr=#1}{#2}
}
\providecommand{\href}[2]{#2}
\begin{thebibliography}{10}

\bibitem{ba11a}
L.~Bader and G.~Lunardon, \emph{Desarguesian spreads}, Ric. Mat. \textbf{60}
  (2011), no.~1, 15--37. \MR{2803932 (2012d:51007)}

\bibitem{et08u}
T.~Etzion and N.~Silberstein, \emph{Error-correcting codes in projective spaces
  via rank-metric codes and {F}errers diagrams}, IEEE Transactions on
  Information Theory \textbf{55} (2009), no.~7, 2909--2919. \MR{MR2589964
  (2010h:94254)}

\bibitem{ga03}
J.~{von zur } Gathen and J.~Gerhard, \emph{Modern computer algebra}, second
  ed., Cambridge University Press, Cambridge, 2003. \MR{MR2001757
  (2004g:68202)}

\bibitem{go11}
E.~Gorla, F.~Manganiello, and J.~Rosenthal, \emph{An algebraic approach for
  decoding spread codes}, arXiv:1107.55230v1 \textbf{[cs.IT]} (2011).

\bibitem{ko03a}
R.~Koetter and M.~Medard, \emph{An algebraic approach to network coding},
  Networking, IEEE/ACM Transactions on \textbf{11} (2003), no.~5, 782 -- 795.

\bibitem{ko08p}
A.~Kohnert and S.~Kurz, \emph{Construction of large constant dimension codes
  with a prescribed minimum distance}, MMICS (Jacques Calmet, Willi Geiselmann,
  and J\"orn M\"uller-Quade, eds.), Lecture Notes in Computer Science, vol.
  5393, Springer, 2008, pp.~31--42.

\bibitem{ko08}
R.~K\"otter and F.R. Kschischang, \emph{Coding for errors and erasures in
  random network coding}, IEEE Transactions on Information Theory \textbf{54}
  (2008), no.~8, 3579--3591.

\bibitem{la01t}
M.~Lavrauw, \emph{Scattered spaces with respect to spreadsand eggs in finite
  projective spaces}, Ph.D. thesis, Eindhoven University of Technology,
  Eindhoven, 2001.

\bibitem{li03}
S.-Y.R. Li, R.W. Yeung, and N.~Cai, \emph{Linear network coding}, Information
  Theory, IEEE Transactions on \textbf{49} (2003), no.~2, 371--381.

\bibitem{li94}
R.~Lidl and H.~Niederreiter, \emph{Introduction to finite fields and their
  applications}, Cambridge University Press, Cambridge, London, 1994, Revised
  edition.

\bibitem{ma08p}
F.~Manganiello, E.~Gorla, and J.~Rosenthal, \emph{Spread codes and spread
  decoding in network coding}, Proceedings of the 2008 IEEE International
  Symposium on Information Theory (Toronto, Canada), 2008, pp.~851--855.

\bibitem{si08a}
D.~Silva, F.R. Kschischang, and R.~K\"otter, \emph{A rank-metric approach to
  error control in random network coding}, Proceedings of the 2008 IEEE
  International Symposium on Information Theory \textbf{54} (2008), no.~9,
  3951--3967.

\bibitem{tr10p}
A.-L. Trautmann, F.~Manganiello, and J.~Rosenthal, \emph{Orbit codes - a new
  concept in the area of network coding}, Information Theory Workshop (ITW),
  2010 IEEE (Dublin, Ireland), August 2010, pp.~1 --4.

\end{thebibliography}
\bibliographystyle{amsplain}

\end{document}